\pdfoutput=1
\documentclass[3p,times]{elsarticle}

\usepackage{ecrc}

\volume{00}

\firstpage{1}

\journalname{}

\runauth{Defu Lin et al.}

\jid{procs}

\jnltitlelogo{}

\CopyrightLine{2023}{Published by Elsevier Ltd.}

\usepackage{amssymb}
\usepackage{amsthm}

\usepackage[figuresright]{rotating}

\usepackage{amsmath}
\usepackage{amsthm}
\usepackage{graphics}
\usepackage{subfigure}
\usepackage{lineno,hyperref}
\usepackage{cases}
\usepackage{bm}
\usepackage{xcolor}
\usepackage{makecell}

\usepackage[linesnumbered,ruled,vlined]{algorithm2e}
\usepackage{hyperref}
\usepackage[noend]{algpseudocode}
\modulolinenumbers[1]

\newtheorem{theorem}{Theorem}

\newtheoremstyle{prestyle}
{0} 
{\topsep} 
{\itshape} 
{} 
{\bfseries} 
{.} 
{.5em} 
{} 
\theoremstyle{prestyle}

\theoremstyle{definition}	

\begin{document}

\begin{frontmatter}



\dochead{}

\title{Patterns Induce Injectivity: A New Thinking in Constructing Injective Local Rules of 1D Cellular Automata over $\mathbb{F}_2$}


\author[1]{Defu Lin}
\ead{2014580@mail.nankai.edu.cn}

\author[2]{Weilin Chen}
\ead{2120220681@mail.nankai.edu.cn}

\author[3]{Chen Wang}
\ead{2120220677@mail.nankai.edu.cn}

\author[4]{Junchi Ma}
\ead{majunchi@mail.nankai.edu.cn}

\author[5]{Chao Wang\corref{cor1}}
\ead{wangchao@nankai.edu.cn}

\address[1,2,3]{Address: College of Software, Nankai University, Tianjin 300350, China}

\cortext[cor1]{Corresponding author.}

\begin{abstract}
 We discovered that certain patterns called injective patterns remain stable during the revolution process, allowing us to create many reversible CA simply by using them to design the revolution rules. By examining injective patterns, we investigated their structural stability during revolutions. This led us to discover extended patterns and pattern mixtures that can create more reversible cellular automata. Furthermore, our research proposed a new way to study the reversibility of CA by observing the structure of local rule $f$. In this paper, we will explicate our study and propose an efficient method for finding the injective patterns. Our algorithms can find injective rules and generate local rule $f$ by traversing $2^{N}$, instead of $2^{2^{N}}$ to check all injective rules and pick the injective ones.
\end{abstract}


\end{frontmatter}


\section{Introduction\label{s1}}
Cellular Automata (CA) are mathematical dynamical systems consisting of a regular network of finite state cells that will change their state simultaneously according to the states of their neighbor under a local rule. As a complicated system, CA can be applied in encryption \cite{abdo2013cryptosystem}, cryptography \cite{nandi1994theory}, data compression \cite{lafe1997data}, image encoding \cite{cappellari2010resolution} and quantum computing \cite{toth2001quantum}, while many significant scenarios require CA to be reversible. 

To study self-reproducing automata, John Von Neumann first introduced Cellular Automata (CA) in \cite{neumann1966theory}, which was published in the 1950s. Hedlund developed the first topological study of reversible cellular automaton in \cite{hedlund1969endomorphisms}, establishing essential concepts for the one-dimensional CA, such as the uniform multiplicity of ancestors and Welch sets. Gregorio and Trautteur, in their 1975 work \cite{DIGREGORIO1975382}, investigated the reversibility, injectivity, and surjectivity of one-dimensional CA and specified the effective condition of the three properties. In the early 1980s, Stephen Wolfram led a new stage in the theoretical research of a cellular automaton by first proposing a one-dimensional elementary cellular automaton (ECA) with a state number of $2$ and a neighborhood radius of $1$ in his work \cite{wolfram1983statistical}. Tome proved necessary and sufficient conditions for null-boundary one-dimensional CA to exit and give his tree-based algorithms in \cite{Tome1994}. In linear CA, Yang proposed a method by constructing the deterministic finite automata (DFA) of CA to solve the reversibility problem in \cite{yang2015reversibility}. In the work  \cite{du2022efficient}, Du improved Yang's method by simplifying each node of DFA and gave a polynomial substitution named standard basis postfix (SBP) algorithm to calculate the reversibility of linear cellular automaton (LCA) under null boundary. Through studying number-conserving CA (NCCA), in their work \cite{WOLNIK2022133075}, Barbara and Maciej proposed a decomposition-based method, which can easily generate a relatively small set of local rules and then verify the reversible ones.

Amoroso and Patt proposed an algorithm to determine whether the global transformation is injective in their 1972 paper \cite{amoroso1972decision} and gave the proof that a nontrivial rule with form $0X1^{(m-1)}0$ is injective, where $m\in Z^{+}$. In 1976's work \cite{bruckner1979garden}, Bruckner proved that the necessary condition for the global transformation to be surjective and injevtive is ``balanced". Based on ``balanced", Maruoka and Kimura gave that the condition of the global transformation of CA of any dimension to be injective is ``hard" \cite{maruoka1976condition}. Klaus Sutner proposed an algorithm to use the de Bruijn diagram to decide the reversibility of cellular automata over the finite field $\mathbb{F}_2$ in 1991, which can reduce the complexity of the injective decision algorithm to quadratic complexity \cite{sutner1991bruijn}. 

However, the former algorithms are not good at quickly generating a batch of CA with injective global transition function due to the traverse process, which limits their computational complexity to the number $2^{2^{N}}$ when given a rule with neighbor size $N$. If the size is getting bigger, the computational cost of these methods is unbearable. For instance, when we need to generate a batch of reversible CA with $7$-neighbor, there will be $2^{2^{7}} = 2^{128}$ rules needed to be traversed over $\mathbb{F}_{2}$ in the worst situation, which is impossible for any computer to complete the calculation in a reasonable time.

To better address the problem, many researchers proposed their methods. Tim uses algebra structure to exhaustive listing reversible one-dimensional CA \cite{BOYKETT2004215} and proposed a method based on tree structure to calculate them. Although they use the pruning method to simplify the tree, their method still needs lots of calculation resources. Juan Carlos introduced a method to use the matrix to study the reversibility of one-dimensional CA which can calculate 6-states 2 neighbor size CA at most in \cite{SECKTUOHMORA2005134}, and based on it, they proposed a new algorithm to calculate a random, reversible cellular automaton with tens of states in paper \cite{seck2017welch}. 

In this paper, we found that if the local rule of CA conforms to some particular patterns, we can easily assert that the local rule is injective. By using these particular patterns for extension, and mixture, we can quickly generate a batch of reversible CA.

We organized the paper as this:
Section \ref{s1} introduces the background of cellular automata.
Then the definition and explanation of relevant basic knowledge are explained in Section \ref{s2}.
Based on previous sections, Section \ref{s3} gives some theorems studied in this paper and their proofs.
The algorithm for constructing nontrivial injective rules we propose in this paper will be explicated in Section \ref{s4}.
We give the results of our algorithm in section \ref{s5}.
Section \ref{s6} concludes and proposes future work.\par

\section{Mathematical description of Cellular Automata\label{s2}}

Usually, we use a quadruple $A = \{d,S,N,f\}$ to describe a CA as follows:
\begin{itemize}
    \item
    $d \in \mathbb{Z}_{+}$ denotes the dimension of the cellular automata space. A \emph{d}-dimensional space is marked $\mathbb{Z}^d$, and each point $\vec{n} \in \mathbb{Z}^{d}$ is called a cell.
    \item
    $S = \{0,1,\cdots,p-1\}$ is a finite set, which denotes the states of cell of CA.
    \item
    $N=(\vec{n}_1, \vec{n}_2,\cdots, \vec{n}_{m}) $ is the neighbor vector of size \emph{m}, where $\vec{n}_{i} \in \mathbb{Z}^{d} $, and $\vec{n}_i \neq \vec{n}_j $  when $i \neq j $ ($1 \leq i, j \leq m$). The \emph{m} neighbors of a cell $ \vec{n} $ can be expressed as $N (\vec{n}) = \vec{n} + \vec{n}_{i} (i=1,2,\cdots, m) $. 
    \item
    $f: S^m\to S$ is the local rule, which maps the current states of all neighbors of a cell to the next state of this cell. 
\end{itemize}

In this paper, we focus on the CA with $d=1$ over $\mathbb{F}_{2}$. We use $c: Z^d \to S$ to denote the configuration of CA, and $c^t$ represents the configuration of CA at time $t$, while $C (C=\{c|c: Z^d \to S\})$ is the set of all configuration. A local rule uniquely defines a global transformation which we denote as $\tau: C \to C$, and then $c^{t+1} = \tau(c^t)$ can denote the mapping process of CA, while $\tau$ is a global transformation. 

For a better explaining, here we apply a Corollary proposed in \cite{Kari2018ReversibleCA}, which enables us to utilize some properties of periodic one-dimensional CA to discuss our method instead of infinite one-dimensional CA. The Corollary is as follows:

we have among one-dimensional cellular automata $\tau \quad injective \iff \tau_{P} \quad injective$, where $\tau_{P}$ denotes the global transformation of periodic one-dimensional CA. 

Based on above, $shift_{k}(c) \iff c$ is obvious. The proof is mechanism moves $k$ steps of the index of the CA, so we leave it to interested readers.

Wolfram proposed a naming method that can succinctly represent a CA rule. We defined the local rule of primary CA as a binary sequence, $f_{i}$ is a binary number which is the result of the corresponding local map, $w$ is the rule number, then $w=\Sigma_{i=0}^{7}f_{i}*2^i$. The $w$ calculated by this rule is called the Wolfram number. For example, the local rule definition table with Wolfram rule number $240$ is shown in Table \ref{tab1}.

\begin{table}[h]
    \centering
    \caption{Local rule definition table with Wolfram number of $11110000$}
    \begin{center}
    \begin{tabular}{c|c|c|c}
      \hline
      Map String & State & Map String & State\\
      \hline
      000 & 0 & 100 & 1 \\
      001 & 0 & 101 & 1 \\
      010 & 0 & 110 & 1 \\
      011 & 0 & 111 & 1 \\
      \hline
    \end{tabular}
    \end{center}
  \label{tab1}
\end{table}

\newtheorem{definition}{Definition}
\begin{definition}
For any $c_1,c_2 \in C $, if $c_1 \neq c_2$, has $\tau(c_1) \neq \tau(c_2)$, then $\tau$ is injective. 
\end{definition}

\begin{definition}
\emph{Trivial Rule} is that for configuration $c^t$, if there is a global transformation $\tau$ determined by a local rule $f$, making $c^{t+1} = \tau(c^t) = shift_{k}(c^t)$, where $k \in \mathbb{Z}, k \neq 0$ represents the shift number of the configuration and if $k < 0$, it means the configuration shift to left, while $k > 0$ means the opposite. Moreover, if the transformation under local rule $f$ changes the corresponding cell state to the complement. The local rule $f$ is trivial as well.
\end{definition}

\begin{theorem}
Global transformation $\tau$ determined by a trivial rule $f$ is injective.
\end{theorem}

The proof here is obvious, so we leave it to the interested readers.

However, the reversibility of CA with nontrivial rules is not apparent. Thus, we need to use an algorithm like Patt has proposed in \cite{amoroso1972decision} to investigate if the rule of CA is injective. Thus, we need to traverse all rules of a given neighbor size until we find enough reversible CA. The only problem is that the method needs to traverse $2^{2^{L+R+1}}$ elements in the worst case. For instance, when the neighbor's size is $7$, the calculation cost can be $2^{2^128} = 2^128$, which is unbearable.

Besides, we can also use the bigraph method proposed in \cite{seck2017welch} to generate some reversible CA. Using this method requires complex settings and may have hidden drawbacks in CA.

We focus on creating reversible CA by defining their local rules and boundary conditions. We discovered that non-trivial reversible CA have stable injective patterns.

\section{The Special Patterns\label{s3}}

We will define the prefix sub-strings and suffix sub-strings of map string $e_{-L}\cdots X\cdots e_R$, where $L, R \in Z^+ $, $L$ is its left radius, $R$ is its right radius, and $L$ and $R$ are not need to be equal.

\begin{definition}
    The \emph{Prefix Sub-strings} of a given map string $e_{-L} \cdots X \cdots e_R$ are $e_{-L}\cdots e_{-L+i}$, where $0 \leq i \leq L+R-1$. The \emph{Suffix Sub-strings} are $e_{R-j} \cdots e_{R}$, where $0 \leq j \leq L+R-1$. Examples of prefix sub-strings and suffix sub-strings are shown in Table \ref{tab2}. In order to simplify our writing, we take $X$ to represent $e_{0}$ and $\overline{e_{0}}$.
\end{definition} 

\begin{table}[h]
    \centering
	\caption{Prefix sub-strings and suffix sub-strings of a given map string.}

	\begin{center}

	\begin{tabular}{c|c}
		\hline
		Prefix sub-strings & Suffix sub-strings \\
		\hline
		$e_{-L}$ & $e_R$ \\
		$e_{-L}e_{-L+1}$ & $e_{R-1}e_R$ \\
        $\cdots$ & $\cdots$ \\
        $e_{-L}\cdots X$ & $X\cdots e_R$ \\
        $e_{-L}\cdots Xe_1$ & $e_{-1}X\cdots e_R$ \\
        $\cdots$ & $\cdots$ \\
        $e_{-L}\cdots X\cdots e_{R-1}$ & $e_{-L+1}\cdots X\cdots e_R$ \\
        \hline
	\end{tabular}
	\end{center}
    \label{tab2}
\end{table}

Since the pattern we built does not limit $L$ and $R$ to be equal, the prefix sub-strings and suffix sub-strings in the table do not correspond to each other in equal length. 

\begin{definition}
    \emph{Injective Pattern} is a map string $e_{-L}\cdots X\cdots e_R$ that every prefix sub-string of the map string, containing $X$, is unequal to its suffix sub-string with the same length, assuming $L \leq R$, while exchange the prefix and the suffix when $L > R$. 
\end{definition} 

For instance, given a map string with $L=1$ and $R=3$, the algorithm proposed in this paper can identify that the injective patterns are $0X011$, $0X110$, $1X001$ and $1X100$.

\begin{definition}
    \emph{Injective Pattern-Induced Rule $f$} is that given an injective pattern $e_{-L}\cdots X\cdots e_R$, and then we confirm two local maps of a local rule $f$ with $L+R+1$ neighbor size and $2^{L+R+1}$ local maps, that are $e_{-L}\cdots e_{0}\cdots e_R \to \overline{e_{0}}$, where $\overline{e_{0}} = 1-e_{0}$ over $\mathbb{F}_{2}$, and the local rule has $e_{-L}\cdots \overline{e_{0}}\cdots e_R \to e_{0}$ correspondingly, while other local maps will keep the form: $a_{-L}\cdots a_{0}\cdots a_R \to a_{0}$. Besides, the rule $f$ must be a ``balance’’ rule either. Moreover, we split the whole local maps (the number is $2^{L+R+1}$) of an injective pattern-induced local rule into two sets. One of which contains injective pattern corresponding local maps is called \emph{Pattern Set}, while the other contains other $2^{L+R+1}-2$ local maps is called \emph{Common Set}. 
\end{definition}

An example of injective pattern-induced local rule is $4278318856$ (Wolfram number, induced by injective patterns $0X011$), which has local maps $00011 \to 0$ and $01011 \to 1$, while other local maps of this rule will keep the corresponding state unchanged.

\begin{theorem}
A revolution local rule $f$ of a one-dimensional cellular automaton is injective if the local rule is induced by an injective pattern.
\end{theorem}
\begin{proof}
Given a configuration $c$ of the CA, introduced in Theorem 3.1, it contains injective patterns $e_{-L} \cdots X \cdots e_R$, as follows:

\begin{equation*}
    c: \cdots a_{-L-1}e_{-L} \cdots X \cdots e_Ra_{R+1} \cdots, 
\end{equation*}
where $a_{i} \in S, i \in \mathbb{N}$ and $-\infty \leq i \leq -L-1$ or $R+1 \leq i \leq \infty$ represents the state of the cell in configuration $c$ that can be $0$ and $1$ over $\mathbb {F}_2$, while $e_{j} \in S, j \in \mathbb{N}$ and $-L \leq j \leq R$ is a specific state of the cell that must be $0$ or $1$. According to the definition of CA, we have map functions:
\begin{equation*}
    \begin{aligned}
        &f: S^{L+R+1} \to S, \\
        &f(a_{-L} \cdots a_{0} \cdots a_{R}) = a_{0}^{'},
    \end{aligned}
\end{equation*}
where $S$ is the state set of cells in the CA, the neighbor size is $L+R+1$ and the revolution local rule is $f$. So we get the map function of injective pattern $$f(e_{-L} \cdots X \cdots e_R) = X^{’},$$ where $X$ represents the cell state $e_{0}$ and $\overline{e_{0}}$, and Table \ref{tab:map_strings_of_IP} shows the map strings, whose corresponding cell (denoted as red) is every element of injective pattern $e_{-L} \cdots $ $X \cdots e_{R}$.

\begin{table}[ht]
    \centering
    \caption{The map strings of every element of injective pattern.}
    \begin{tabular}{c|c}
    \hline
        Map strings of elements before $X$ & Map strings of elements after $X$ \\
    \hline
        $a_{-L-1}e_{-L} \cdots \textcolor{red} {e_{-1}} \cdots e_{R-1}$ 
        & $e_{-L+1} \cdots \textcolor{red}{e_{1}} \cdots e_{R}a_{R+1}$  \\
        $a_{-L-2}a_{-L-1}e_{-L} \cdots \textcolor{red}{e_{-2}} \cdots e_{R-2}$ 
        & $e_{-L+2} \cdots \textcolor{red}{e_{2}} \cdots e_{R}a_{R+1} a_{R+2}$  \\
        $\cdots$ 
        & $\cdots$ \\
        $a_{-2L} \cdots a_{-L-1}\textcolor{red}{e_{-L}} \cdots e_{R-L}$ 
        & $e_{-L+R} \cdots \textcolor{red}{e_{R}} a_{R+1} \cdots a_{2R}$ \\
    \hline
    \end{tabular}
    \label{tab:map_strings_of_IP}
\end{table}

Since the revolution local rule $f$ is an injective pattern-induced rule, we have the map functions as follows.
\begin{equation*}
    \begin{aligned}
        &f(a_{-L} \cdots a_{0} \cdots a_{R}) = a_{0}^{'}, \\
        &a_{0}^{'} = \begin{cases}
            \overline{a_{0}}, & when \quad a_{-L} \cdots a_{0} \cdots a_{R} = e_{-L} \cdots X \cdots e_R \\
            a_{0}, & others
        \end{cases},
    \end{aligned}
\end{equation*}
where $a_{-L} \cdots a_{0} \cdots a_{R}$ is the map string we have introduced previously, that is the map strings shown in the Table \ref{tab:map_strings_of_IP} and the injective pattern $e_{-L} \cdots X \cdots e_R$. So, we have the map functions of these map strings as follows:
\begin{equation*}
    \begin{aligned}
        &f(a_{-2L} \cdots a_{-L-1}e_{-L} \cdots e_{R-L}) = e_{-L}^{'}, \\
        &\cdots \\
        &f(a_{-L-1}e_{-L} \cdots e_{-1} \cdots e_{R-1}) = e_{1}^{'}, \\
        &f(e_{-L} \cdots e_{0} \cdots e_{R}) = e_{0}^{'}, \\
        &f(e_{-L+1} \cdots e_{1} \cdots e_{R}a_{R+1}) = e_{1}^{'}, \\  
        &\cdots \\
        &f(e_{-L+R} \cdots e_{R}a_{R+1} \cdots a_{2R}) = e_{R}^{'}.
    \end{aligned}
\end{equation*}

In addition, the map strings in the Table \ref{tab:map_strings_of_IP} will not equal to the injective pattern $e_{-L} \cdots X \cdots e_{R}$ no matter what value $a_{i}$ takes, because of the definition of injective pattern. Thus, we get $e_{j}^{’} = e_{j}$, where $j \in \mathbb{N}$ and $-L \leq j \leq R$, except for $e_{0}^{’}$ which is $\overline{e_{0}}$. That is:
\begin{equation*}
    \begin{aligned}
        c^{t} &: \cdots a_{-L-1}e_{-L} \cdots e_{0} \cdots e_Ra_{R+1} \cdots, \\
        c^{t+1} = \tau(c) &: \cdots a_{-L-1}e_{-L} \cdots \overline{e_{0}} \cdots e_Ra_{R+1} \cdots, \\
        c^{t+2} = \tau(\tau(c^t)) = c^t &: \cdots a_{-L-1}e_{-L} \cdots e_{0} \cdots e_Ra_{R+1} \cdots,
    \end{aligned}
\end{equation*}
where $c$ represents the configuration of CA, while $t$ is time $t$ and so as $t+1$, $t+2$, and $\tau$ is the global transformation of CA. $f$ is injective and the total number configuration of the CA with injective pattern-induced local rule $f$ is $2$, which is a notable feature.
\end{proof}

\begin{definition}
    \emph{Extended Pattern} can be generated by adding $k$ elements to the left end of injective pattern and $h$ elements to the right end of injective pattern $e_{-L} \cdots X \cdots e_{R}$, where $k, h \in \mathbb{N}$.
\end{definition}

To better explain, we noted each element as $a$ that can represent all states in $S$ and the given injective pattern is $e_{-L} \cdots X \cdots e_{R}$, then we have $a_{-L-k}\cdots a_{-L-1}$ $e_{-L} \cdots X \cdots e_{R}$ $a_{R+1}\cdots a_{R+h}$, that is extended pattern. Since $h,k\in \mathbb{N}$, the extended pattern can be any length, we specify a local rule diameter $D = L+R+1+k+h$ here to assist us in studying the rule $f$.

We give a specific example here. A extended pattern based on injective pattern $0X011$ with $k=1, h=2$ is $a0X011aa$, which represents $0001100$, $0011100$, $00001101$, $00101101$, $00001110$, $00101110$, $00001111$, $00101111$, $10001100$, $10101100$, $10001101$, $10101101$, $10001110$, $10101110$, $10001111$, $10101111$, and will in the pattern set while others map strings will in the common set.

Similar to the injective pattern-induced local rule, we can also induce the corresponding local rule $f$ using the extended pattern and get the pattern set and the common set.
\begin{definition}
    \emph{Extended Pattern-Induced Rule $f$} is that as long as the map string denoted as $a_{1} \cdots X \cdots a_{D}$ equals to extended pattern $a_{-L-k}\cdots a_{-L-1}$ $e_{-L} \cdots X \cdots e_{R}$ $a_{R+1}\cdots a_{R+h}$, have $f(a_{-L-k}\cdots a_{-L-1}$ $e_{-L} \cdots X \cdots e_{R}$ $a_{R+1}\cdots a_{R+h}) = \overline{x}$, where $X$ represents $e_{0}$ or $\overline{e_{0}}$. otherwise the state will be still, that is $f(a_{1} \cdots X \cdots a_{D}) = X$.
\end{definition}

For example, given an extended pattern $10X1a$, which is extended by injective pattern $10X1$, the Wolfram number of extended pattern-induced local rule $f$ is $1007612144$, corresponding to binary number $111100000011101111000011110000$. If and only if the map string is equal to extended pattern $10X1a$, the state of $X$ will change after mapping, where $X, a \in \{0,1\}$, so we have $10X10 \to \overline{X}, 10X11 \to \overline{X}$. Otherwise, the state will be stable.
\begin{theorem}
    An extended pattern-induced local rule $f$ of the CA is injective.
\end{theorem}
\begin{proof}
    Similar to the proof of Theorem 3.1, we first define a CA with an extended pattern-induced rule $f$. As we introduced previously, the neighbor size is $D$ and the map functions are as follows:

\begin{equation*}
    \begin{aligned}
        &f: S^{D} \to S, \\
        &f(a_{1} \cdots X \cdots a_{D}) = X^{'},
    \end{aligned}
\end{equation*}
where $S$ is the state set of cells, and the neighbor size is $D$. The map string $a_{1} \cdots X \cdots a_{D}$ is equal to $a_{-L-k} \cdots X \cdots a_{R+h}$, and $e_{-L} \cdots X \cdots e_{R}$ is injective pattern, if and only if $a_{-L-k} \cdots X \cdots a_{R+h} = a_{-L-k} \cdots a_{-L-1}e_{-L} \cdots X \cdots e_{R} a_{R+1} \cdots a_{R+h}$, we get a map function:
\begin{equation*}
    \begin{aligned}
        &f(a_{-L-k} \cdots a_{-L-1}e_{-L} \cdots X \cdots e_{R}a_{R+1} \cdots a_{R+h}) = e_{0}^{'}, \quad e_{0}^{'} = \overline{e_{0}}. \\ 
    \end{aligned}
\end{equation*}

According to the extended pattern-induced rule and Theorem 3.1, we know the result. Then, the results of the map functions of the CA are as follows: 
\begin{equation*}
    \begin{aligned}
        &f(a_{-2L-k} \cdots a_{-L-1} e_{-L} \cdots e_{R} a_{R+1} \cdots a_{R+h-L}) = e_{-L}^{'},\\
        & \cdots \\
        &f(a_{-L-k-1} \cdots a_{-L-1} e_{-L} \cdots e_{-1} \cdots e_{R}a_{R+1} \cdots a_{R+h-1}) = e_{-1}^{'}, \\
        &f(a_{-L-k} \cdots a_{-L-1}e_{-L} \cdots e_{0} \cdots e_{R} a_{R+1} \cdots a_{R+h}) = e_{0}^{'}, \\
        &f(a_{-L-k+1} \cdots a_{-L-1}e_{-L} \cdots e_{1} \cdots e_{R}a_{R+1} \cdots e_{R+h+1}) = e_{1}^{'}, \\
        & \cdots \\
        &f(a_{-L-k+R} \cdots a_{-L-1}e_{-L} \cdots e_{R} a_{R+1} \cdots a_{2R+h}) = e_{R}^{'}. \\
    \end{aligned}
\end{equation*}

Since every map string shown before is unequal to extended pattern $a_{-L-k}\cdots $ $ a_{-L-1}e_{-L} \cdots e_{0}$ $ \cdots e_{R} a_{R+1} \cdots a_{R+h}$, no matter what value $a_{i}$ takes, we have the results that $e_{j}^{’} = e_{j}$, where $j \in \mathbb{N}$ and $-L \leq j \leq R$, except $e_{0}^{’} = \overline{e_{0}}$. We have proof that even in the extended pattern, the injective pattern part can keep its structure in mapping. 

In the next step, we discuss the extended part of the extended pattern. The extended parts have two situations, which are 1) strings that can formulate injective pattern $a_{-L-k} \cdots a_{-L-1}e_{-L} \cdots X \cdots e_{R}$ $a_{R+1} \cdots a_{R+h}$ and 2) strings that can not formulate injective pattern. In the first situation, we can have the same solution as $f(a_{-L-k} \cdots a_{-L-1}e_{-L} \cdots e_{0} \cdots e_{R} $ $ a_{R+1} \cdots a_{R+h}) = \overline{e_{0}}$ as long as we take the extended pattern as a window and move it to the corresponding position, and due to Definition 3, the injective pattern will always keep its structure, while in the second situation, we have the solution shown as previous excepts the solution mentioned in the first situation, that is, the results are $e_{j}^{’} = e_{j}$ in all time, except for $e_{0}^{’}$ which is $\overline{e_{0}}$ mentioned in 1). 

As a result, we have the conclusion that $\tau(\tau(c)) = shifte_{k}(c)$, and apparently, the local rule $f$ is injective.
\end{proof}
So far, we introduced two patterns that can induce injective local rules respectively. We discovered a unique property that allows them to maintain the "patterns" structure during the mapping process. The state of the corresponding cell can only change in the next cycle if the local map matches the "patterns.". So, we define this property as follows, and discuss a new way to generate more injective rules.

\begin{definition}
    \emph{Pattern Independence} is that given a set of injective patterns, 1) for the patterns with the same length, the prefix sub-strings of each pattern, containing $X$, are unequal to the corresponding length suffix sub-strings of all other patterns. 2) For the patterns with different lengths, it should conform to situation 1) at first, and also require the shorter injective patterns not equal to any of the same length sub-strings of the longer injective patterns.
\end{definition}

To better explain, in situation 1), take $e^{j}_{-L-k} \cdots X \cdots e^{j}_{R+h}, j = 1,2,\cdots, k$ to represent the injective patterns, the property is that:
\begin{equation*}
    \begin{aligned}
        &e^{j}_{-L-k} \cdots X \cdots e^{j}_{r} \neq \begin{cases}
            e^{1}_{R+h-(r+L+k+1)} \cdots e^{1}_{R+h}, \\
            \cdots \\
            e^{j-1}_{R+h-(r+L+k+1)} \cdots e^{j-1}_{R+h} \\
            e^{j+1}_{R+h-(r+L+k+1)} \cdots e^{j+1}_{R+h} \\
            \cdots \\
            e^{k}_{R+h-(r+L+k+1)} \cdots e^{k}_{R+h} \\
        \end{cases},
    \end{aligned}
\end{equation*}
where $r\in \mathbb{Z}, 0 \leq r < R+h$, and assume $R+h \geq L+k$ (if $R+h \leq L+k$, the “$\neq$” situation is vice versa).

In situation 2), take $e^{j}_{-L1} \cdots X \cdots e^{j}_{R1}$ to be the longer injective patterns and $e^{j}_{-L2} \cdots X \cdots e^{j}_{R2}$ to be the shorter injective patterns, where $j = 1,2,\cdots, k$ and $L1, L2, R1, R2 \in \mathbb{Z}, 0 \leq L1, L2 \leq L+k, 0 \leq R1, R2 \leq R+h$, and that is as follows:
\begin{equation*}
    \begin{aligned}
        &e^{j}_{-L1} \cdots X \cdots e^{j}_{r} \neq e^{j_1}_{L2-(L1+r+1)} \cdots e^{j_1}_{L2}, \quad j_1 = 1,2,\cdots,k, j_1 \neq j, \\
        & e^{j}_{-L1} \cdots X \cdots e^{j}_{R1} \neq 
        Substring(e^{j_1}_{-L2} \cdots X \cdots e^{j_1}_{R2}), \quad j_1 = 1,2,\cdots,k, j_1 \neq j,
    \end{aligned}
\end{equation*}
where $r\in \mathbb{Z}, 0 \leq r < R1$, and assume $R1 \geq L1$ (if $R1 \leq L1$, the “$\neq$” situation is vice versa), and $j_1, j_1 =1,2,\cdots,k, j_1 \neq j$ means we traverse $k-1$ patterns. $Substring$ means to get the substring with the same length as the shorter injective patterns. Longer and shorter injective patterns have the same length, left radius, and right radius, and satisfy pattern independence.

According to \textbf{Definition 7}, we can consider the injective patterns and extended patterns to be self-independence.

\begin{definition}
    \emph{Patterns Mixture} is a pattern set including injective patterns and extended patterns that satisfy pattern independence. We use \emph{Mixture Set} to denote the given set. 
\end{definition} 

Assume that the rule diameter is $D$, and $D$ is the longest length of all patterns. Then, the patterns mixture includes situations as follows.

\begin{itemize}
    \item [1)] All elements of the mixture are \textit{injective patterns} with $D$ length.

    \item [2)] The mixture set consists of \textit{injective patterns} and \textit{extended patterns}, and all of them are of $D$ length.

    \item [3)] All elements of the mixture set are \textit{extended patterns} with a diameter $D$, which are extended by injective patterns with a diameter smaller than $D$.
\end{itemize}

Similar to the previous, we can induce local rules based on the patterns mixtures. The definition of patterns mixture-induced rule $f$ is as follows.
\begin{definition}
    \emph{Patterns Mixture-Induced Rule $f$}, of $D$ length and with $2^{D}$ local maps, is that if and only if the map strings of local maps in the mixture set, the state of $X$ will change during the mapping process, that is $e_{1}\cdots X\cdots e_{D} \to X^{'}$, where $e_{1}\cdots X\cdots e_{D} \in Mixtrue Set$ and $X^{'}$ represents the changed state of the corresponding cell.
\end{definition}

\begin{theorem}
     A patterns mixture induced rule $f$ of the CA is injective.
\end{theorem}

\begin{proof}
    
Assume there are $k, k \in \mathbb{Z}^{+}$ patterns with different lengths in the mixture set, while the max length $D$. The patterns in the mixture set conform to pattern independence, and according to Definition 8, this is \textit{ patterns mixture}. Given a one-dimensional CA with patterns mixture induced rule $f$ will have map functions:
\begin{equation*}
    \begin{aligned}
        &f: S^{D} \to S, \\
        &f(a_{1} \cdots X \cdots a_{D}) = X^{'},
    \end{aligned}
\end{equation*}
where $D$ is the neighbor size, $S$ is the state set, and $a_{1} \cdots X \cdots a_{D}$ represents all possible states of map strings with a given length, that is $a_{1} \cdots X \cdots a_{D} \in S^{D}$, while $f$ is local rule. Based on the different situations of elements in the mixture set, we give our proof respectively, as follows.

\begin{itemize}
    \item[1)] Every element in the mixture set is an injective pattern with $D$ neighbor size, denoted as $e^{j}_{-L-k} \cdots X \cdots e^{j}_{R+h}, j = 1,2,\cdots,k$. 

    \item[2)] The elements in mixture set are all $D$ size,  including injective patterns and extended patterns. We denote the injective patterns as $e^{j}_{-L-k} \cdots X \cdots e^{j}_{R+h}, j = 1,2,\cdots,k$. and since the extended part of the extended patterns can represent any state, we only need to consider the injective part, which we proved in Theorem 3.2, and denoted them as $e^{j}_{-L} \cdots X \cdots e^{j}_{R}, j = 1,2,\cdots,k$.

    \item[3)] The elements in the mixture set are all extended patterns with size $D$.
\end{itemize}

For \textbf{situation 1)}, in the mixture set, there are $k$ injective patterns that satisfy patterns independence, so we have the map functions as follows:
\begin{equation*}
    \begin{aligned}
        &f(a_{-L-k} \cdots a_{0} \cdots a_{R+h}) = a_{0}^{'}, \\  
        &a_{0}^{'} = \begin{cases}
            \overline{a_{0}}, & a_{-L-k} \cdots a_{0} \cdots a_{R+h} \in Mixture Set \\
            a_{0}, & others
        \end{cases},
    \end{aligned}
\end{equation*}
where $a_{-L-k} \cdots a_{0} \cdots a_{R+h}$ is a string window with $D$ length on the configuration $c$ of the given CA, shown in Fig. \ref{fig:Injective_patterns_in_PM}. The blue box is the window and the red cell represents the corresponding cell of this window. 

\begin{figure}[ht]
    \centering
    \includegraphics[scale=0.4]{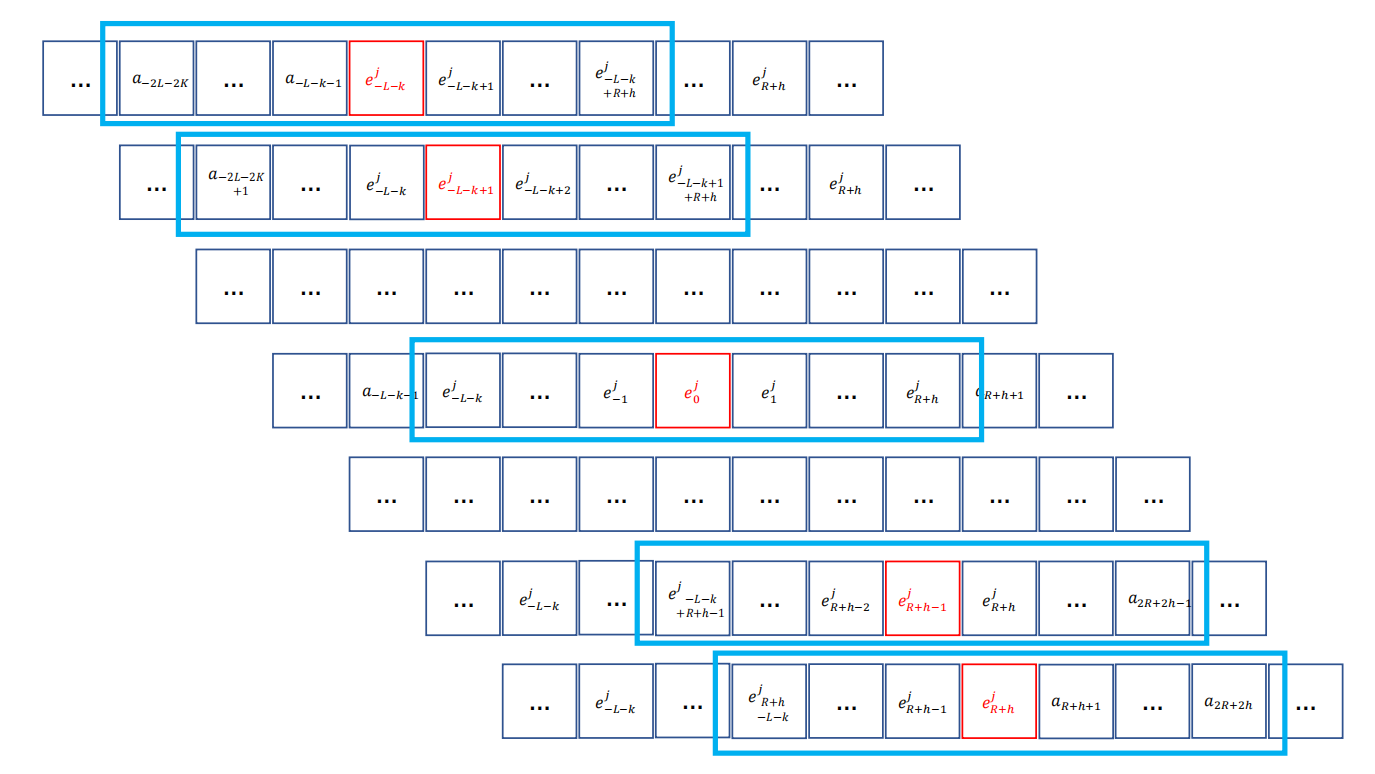}
    \caption{The window of given CA with patterns mixture induced rule $f$ while all elements in mixture set are injective patterns.}
    \label{fig:Injective_patterns_in_PM}
\end{figure}

According to the definition, as long as the string in the window equals to $e^{j}_{-L-k} \cdots $ $ X \cdots e^{j}_{R+h}, j = 1,2,\cdots,k$, where $X$ can represent both $e^{j}_{0}$ and $\overline{e}^{j}_{0}$, which is the injective pattern in the mixture set, the state of the corresponding cell can change after mapping, or the state will be still. 

In addition, no matter what value the $a_{i}$ take, the content in the window will not equal to $e^{j}_{-L-k} \cdots $ $ X \cdots e^{j}_{R+h}, j = 1,2,\cdots,k$, where $i \in \mathbb{Z}, -2L-2k \leq i \leq -L-k-1 \quad or \quad R+h+1 \leq i \leq 2R+2h$, so in the Fig. \ref{fig:Injective_patterns_in_PM}, only one window can change the state of the corresponding cell, and this situation can be extended to the whole configuration $c$ and all injective patterns, which are in the mixture set. That is, the injective patterns will keep their structure always. In the next time, the situation is the same. So, we have $\tau(\tau(c)) = shift_{k}(c)$, which has proven the patterns mixture induced rule $f$ is injective.

For \textbf{situation 2)}, the mixture set will have both injective patterns and extended patterns. However, for extended patterns, we only need to consider the effectiveness between the shorter injective patterns which construct the extended patterns, and the longer injective patterns. Similar to the \textbf{situation 1)}, we define a window with the same length with the shorter injective patterns, and the corresponding cell will just on the same situation as the first cell of the longer injective patterns, and then move the window one cell a step to traverse all cell of injective patterns, shown as Fig. \ref{fig:extended_patterns_in_PM}. 

\begin{figure}[ht]
    \centering
    \includegraphics[scale=0.2]{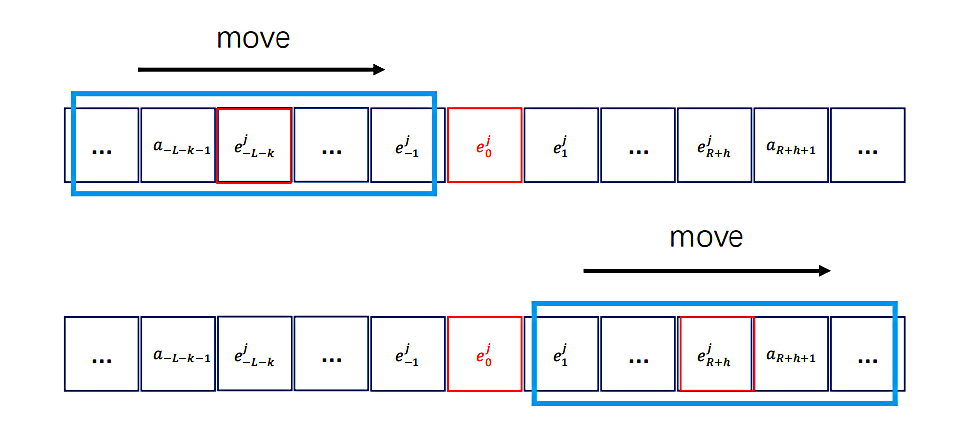}
    \caption{The window of given CA with patterns mixture induced rule $f$ while the mixture set includes both injective patterns and extended patterns.}
    \label{fig:extended_patterns_in_PM}
\end{figure}

According to \textbf{Definition 8}, the content of the window in any step will have no chance to equal to the shorter injective patterns, which means that the shorter injective patterns can not break the structure of the longer injective patterns, while if do the same operation with the window change to the length of the longer injective patterns and the corresponding cell in the first cell of the shorter injective patterns, the longer injective patterns can not break the structure of the shorter injective patterns as well. In addition, we have proven injective patterns with the same length can not influence each other to break their structure in \textbf{situation 1)}.

To sum up, firstly, given a CA with patterns mixture induced rule $f$, and define numerous windows with different lengths that can represent all lengths of the injective patterns with $D$ length or injective patterns which construct the $D$ length extended patterns in the mixture set. Then, use the windows to traverse the initial configuration $c$ of the CA, if the content of the windows is equal to any injective patterns mentioned latest, according to Definition 9, the state of the corresponding cell of the window will change in the next time, while the structure of mentioned latest injective patterns will remain, since all of them will not break the structure of others which we have proven previously. Do the same operation again, the changed state cell will change its state as well, that is, the same as the initial configuration. So, we have $\tau(\tau(c)) = shift_{k}(c)$, and apparently, the local rule $f$ is injective.

For \textbf{situation 3)} the proof is the same as \textbf{situation 2)}, which only needs to change the length of the window, and the conclusion is the same as \textbf{situation 2)} too. Thus, in all situations, the local rule $f$ will ensure the configuration of the CA has $\tau(\tau(c)) = shift_{k}(c)$, so the local rule $f$ is injective. 
\end{proof}

\section{Algorithms\label{s4}}
Firstly, we give an algorithm to check the pattern independence named Algorithm \ref{PID}.

\begin{algorithm}[H]
    \caption{Patterns Independence Detection}
    \label{PID}
    \SetKwData{True}{True}
    \SetKwData{False}{False}
    \SetKwFunction{PrefixSubstring}{PrefixSubstring}
    \SetKwFunction{SuffixSubstring}{SuffixSubstring}
    \SetKwFunction{Substring}{Substring}
    
    \SetKwInOut{Input}{input}
    \SetKwInOut{Output}{output}
    \Input{$PatternShort$(injective pattern), $L_1, R_1$(left radius and right radius); \\ $PatternLong$(injective pattern), $L_2, R_2$(left radius and right radius)}
    \Output{\True or \False}
    \If{$PatternShort$ equals $PatternLong$}{
        $DetectionType \gets InjectivePattern$;
    }
    \For{$length \gets L_1$ \KwTo $L_1 + R_1$}{
        \If{\PrefixSubstring{$PatternShort$, $length$} equals \SuffixSubstring{$PatternLong$, $length$}}{
            return \False;
        }
    }
    \eIf {$DetectionType$ is $InjectivePattern$}{
        return \True
    }{
        \For {$begin \gets -L_2$ \KwTo $R_2 - (L_1 + R_1 +1)$} {
            \If {$PatternShort$ equals \Substring {$PatterLong$, $begin$, $begin + L_1+R_1+1$}} {
                return \False
            }
        }
        return \True
    }
\end{algorithm}

In Algorithm \ref{PID}, the function $PrefixSubstring$ and $SuffixSubsring$ are defined in \textbf{Definition 3}. The first augment is the pattern, and the second augment is the length. Furthermore, the function $Substring$ has $3$ augments as input. The first one is a pattern, the second one is the beginning position and the last one is the end position.

\textbf{Complexity Analysis}. From Algorithm \ref{PID}, we get that the main cost will be the two ``$for$'' loops and the comparison process, concerning the length of input patterns. Denote the length of patterns to be $n$, then the complexity will be $O(n^2)$ in the worst situation.

Secondly, we give an algorithm to generate all injective patterns of a given length. 

\begin{algorithm}
    \caption{Injective Patterns Generation}
    \label{IPG}
    \SetKwData{InjectivePatternSet}{InjectivePatternSet}
    \SetKwData{True}{True}
    \SetKwData{False}{False}
    \SetKwData{MapString}{MapString}
    \SetKwFunction{PID}{PID}
    \SetKwFunction{RMTToMapstring}{RMTToMapstring}
    \SetKwInOut{Input}{input}
    \SetKwInOut{Output}{output}
    \Input{$L$(left radius), $R$(right radius)}
    \Output{\InjectivePatternSet}
    \InjectivePatternSet $\gets \varnothing$ \\
    \For {$RMT \gets 0$ \KwTo $2^{L+R+1} - 1$} {
        \MapString $\gets$ \RMTToMapstring{$RMT$} \\
        \If {\PID{\MapString, \MapString} is \True} {
            \InjectivePatternSet $\gets$ \InjectivePatternSet $+ \{$ \MapString $\}$
        } 
    }
    return \InjectivePatternSet
\end{algorithm}

In Algorithm \ref{IPG}, the function $PID$ is defined in the Algorithm \ref{PID}.

\textbf{Complexity Analysis}. The complexity of the function $PID$ is $O(n^2)$. Since we define $n$ as the length of the pattern, we know $n = L + R + 1$ and the function $IPG$, defined in Algorithm \ref{IPG}, is an ``$for$'' loop with an exponential number, so the complexity is $O(n^{2} 2^{n})$.

Then we can use the results of Algorithm \ref{IPG} to generate injective pattern-induced rules, while each element of the results set can generate an injective rule $f$. The algorithm here can be easily designed as follows.

\begin{algorithm}
    \caption{Injective Rule Generation}
    \label{IRG}
    \SetKwData{Pattern}{Pattern}
    \SetKwData{InjectiveRule}{InjectiveRule}
    \SetKwData{MapString}{MapString}
    \SetKwData{LocalMap}{LocalMap}
    \SetKwFunction{RMTToMapstring}{RMTToMapstring}
    \SetKwFunction{Substring}{Substring}
    \SetKwFunction{CharToNumber}{CharToNumber}
    \SetKwInOut{Input}{input}
    \SetKwInOut{Output}{output}
    \Input{\Pattern, $L$(left radius), $R$(right radius)}
    \Output{\InjectiveRule $f$}
    $f \gets \varnothing$ \\
    \For {$RMT \gets 0$ \KwTo $2^{L+R+1} - 1$} {
        \MapString $\gets$ \RMTToMapstring{$RMT$} \\
        \eIf {\MapString equals \Pattern} {
            \LocalMap $\gets$ \MapString : $1 -$ \CharToNumber{\Substring{\MapString, 0, 0}}
        } {
            \LocalMap $\gets$ \MapString : \CharToNumber{\Substring{\MapString, 0, 0}}
        }
        $f$ $\gets$ $f + \{$ \LocalMap $\}$
    }
    return \InjectiveRule $f$
\end{algorithm}

\textbf{Complexity Analysis}. The main cost is the traverse with exponential number and the comparison process. Thus, the complexity is $O(n2^{n})$, and we can see from the Algorithm \ref{IRG} that each injective pattern can generate an injective rule. The function $CharToNumber$ can turn a char into an integer number.

According to \textbf{Definition 6}, we can easily extend the injective patterns we get by Algorithm \ref{IPG}, and the extended pattern-induced injective rules can be generated by Algorithm \ref{IRG} too. Besides, by \textbf{Definition 9}, the extend patterns and other injective patterns with the same length of extend patterns, can generate a combinatorial number of patterns mixtures and generate a corresponding number of injective rules. The Algorithms of this generation can be realized by simply modifying the algorithms we designed previously.

\section{Experiment\label{s5}}
Based on the algorithms we proposed, we calculate the number of injective patterns and extended patterns with neighbor size from $3$ to $10$. The results are shown in Table \ref{tab:injective_extended_res}. 

\begin{table}[ht]
    \centering
    \caption{The number of results with different neighbor size.}
    \label{tab:injective_extended_res}
    \begin{tabular}{c|c|c|c}
    \hline
        \makecell[c]{Neighbor size} & 
        \makecell[c]{Injective rules for \\ injective patterns}  &
        \makecell[c]{Injective rules for \\ extended patterns} & 
        \makecell[c]{Total injective rules \\ (except trivial)} \\ 
    \hline  
        3 & 0 & 0 & 0 \\
        4 & 4 & 0 & 4 \\
        5 & 14 & 8 & 26 \\
        6 & 52 & 40 & Can't calculate \\
        7 & 148 & 162 & Can't calculate \\
        8 & 408 & 528 & Can't calculate \\
        9 & 1040 & 1562 & Can't calculate \\
        10 & 2556 & 4268 & Can't calculate \\
    \hline
    \end{tabular}
\end{table}

As we can see, we can specifically calculate the number of injective patterns and extended patterns in different neighbor sizes. In neighbor size $4$, there will be $4$ injective rules (except trivial rules), and in neighbor size $5$, the number will be $26$ (except trivial rules), while our methods can find $4$ injective rules in neighbor size $4$ and $22$ injective rules in neighbor size $5$. When it comes to neighbor size $6$ or more, in traditional algorithms that aim to find a total number of injective rules, the amount of calculation explodes to $2^{64}$ in neighbor size $6$ and $2^{128}$ in neighbor size $7$, so we can not get the final results.

Besides, we only list the injective patterns and the extended patterns here, since patterns mixture can be a combination problem with number based on the total number of injective patterns and extended patterns. So, the actual number of injective rules our methods can find will be greater than what we show in Table \ref{fig:Injective_patterns_in_PM}.

\section{Conclusion and Future Work\label{s6}}
We proposed a patterns-based method that can efficiently generate injective local rule $f$, which constrains the number of configurations of the CA to be $2$. In our method, we simply design the elements structure of the mapping string and prove the structure is stable in any mapping process, which enables us to induce injective local rule $f$. Besides, the global transformation defined by the rule $f$ holds $\tau(\tau(c)) = shift_{k}(c)$, which means the CA with the rule $f$ will only have $2$ configuration, so the CA is reversible. 

Although the CA we discuss are one-dimensional CA over $\mathbb{F}_2$, we can easily extend our method to CA over $\mathbb{F}_p$ or even two-dimensional CA, which can be our future work. Furthermore, we found all ``patterns'' induced rules have defined $\tau(\tau(c)) = \tau^{2}(c) = shift_{k}(c)$, which constrains the number of configurations of CA to be $2$, but there are many reversible CA with multi-configuration more than $2$. Thus, another future work of ours is to extend the number $2$ to be $k$ which can include all injective situations. It is worth mentioning that implementing ``Patterns’’ to study the reversibility of CA is a brand new idea, which can help us efficiently generate a lot of reversible CA. 





\bibliographystyle{elsarticle-harv}
\bibliography{fbref}

\begin{thebibliography}{21}
\expandafter\ifx\csname natexlab\endcsname\relax\def\natexlab#1{#1}\fi
\providecommand{\url}[1]{\texttt{#1}}
\providecommand{\href}[2]{#2}
\providecommand{\path}[1]{#1}
\providecommand{\DOIprefix}{doi:}
\providecommand{\ArXivprefix}{arXiv:}
\providecommand{\URLprefix}{URL: }
\providecommand{\Pubmedprefix}{pmid:}
\providecommand{\doi}[1]{\href{http://dx.doi.org/#1}{\path{#1}}}
\providecommand{\Pubmed}[1]{\href{pmid:#1}{\path{#1}}}
\providecommand{\bibinfo}[2]{#2}
\ifx\xfnm\relax \def\xfnm[#1]{\unskip,\space#1}\fi
\bibitem[{Abdo et~al.(2013)Abdo, Lian, Ismail, Amin and
  Diab}]{abdo2013cryptosystem}
\bibinfo{author}{Abdo, A.}, \bibinfo{author}{Lian, S.},
  \bibinfo{author}{Ismail, I.A.}, \bibinfo{author}{Amin, M.},
  \bibinfo{author}{Diab, H.}, \bibinfo{year}{2013}.
\newblock \bibinfo{title}{A cryptosystem based on elementary cellular
  automata}.
\newblock \bibinfo{journal}{Communications in Nonlinear Science and Numerical
  Simulation} \bibinfo{volume}{18}, \bibinfo{pages}{136--147}.
\bibitem[{Amoroso and Patt(1972)}]{amoroso1972decision}
\bibinfo{author}{Amoroso, S.}, \bibinfo{author}{Patt, Y.N.},
  \bibinfo{year}{1972}.
\newblock \bibinfo{title}{Decision procedures for surjectivity and injectivity
  of parallel maps for tessellation structures}.
\newblock \bibinfo{journal}{Journal of Computer and System Sciences}
  \bibinfo{volume}{6}, \bibinfo{pages}{448--464}.
\bibitem[{Boykett(2004)}]{BOYKETT2004215}
\bibinfo{author}{Boykett, T.}, \bibinfo{year}{2004}.
\newblock \bibinfo{title}{Efficient exhaustive listings of reversible one
  dimensional cellular automata}.
\newblock \bibinfo{journal}{Theoretical Computer Science}
  \bibinfo{volume}{325}, \bibinfo{pages}{215--247}.
\newblock \URLprefix
  \url{https://www.sciencedirect.com/science/article/pii/S0304397504003925},
  \DOIprefix\doi{https://doi.org/10.1016/j.tcs.2004.06.007}.
  \bibinfo{note}{theoretical Aspects of Cellular Automata}.
\bibitem[{Bruckner(1979)}]{bruckner1979garden}
\bibinfo{author}{Bruckner, L.K.}, \bibinfo{year}{1979}.
\newblock \bibinfo{title}{On the garden-of-eden problem for one-dimensional
  cellular automata}.
\newblock \bibinfo{journal}{Acta Cybernetica} \bibinfo{volume}{4},
  \bibinfo{pages}{259--262}.
\bibitem[{Cappellari et~al.(2010)Cappellari, Milani, Cruz-Reyes and
  Calvagno}]{cappellari2010resolution}
\bibinfo{author}{Cappellari, L.}, \bibinfo{author}{Milani, S.},
  \bibinfo{author}{Cruz-Reyes, C.}, \bibinfo{author}{Calvagno, G.},
  \bibinfo{year}{2010}.
\newblock \bibinfo{title}{Resolution scalable image coding with reversible
  cellular automata}.
\newblock \bibinfo{journal}{IEEE Transactions on Image Processing}
  \bibinfo{volume}{20}, \bibinfo{pages}{1461--1468}.
\bibitem[{{Di Gregorio} and Trautteur(1975)}]{DIGREGORIO1975382}
\bibinfo{author}{{Di Gregorio}, S.}, \bibinfo{author}{Trautteur, G.},
  \bibinfo{year}{1975}.
\newblock \bibinfo{title}{On reversibility in cellular automata}.
\newblock \bibinfo{journal}{Journal of Computer and System Sciences}
  \bibinfo{volume}{11}, \bibinfo{pages}{382--391}.
\newblock \URLprefix
  \url{https://www.sciencedirect.com/science/article/pii/S0022000075800596},
  \DOIprefix\doi{https://doi.org/10.1016/S0022-0000(75)80059-6}.
\bibitem[{Du et~al.(2022)Du, Wang, Wang and Gao}]{du2022efficient}
\bibinfo{author}{Du, X.}, \bibinfo{author}{Wang, C.}, \bibinfo{author}{Wang,
  T.}, \bibinfo{author}{Gao, Z.}, \bibinfo{year}{2022}.
\newblock \bibinfo{title}{Efficient methods with polynomial complexity to
  determine the reversibility of general 1d linear cellular automata over zp}.
\newblock \bibinfo{journal}{Information Sciences} \bibinfo{volume}{594},
  \bibinfo{pages}{163--176}.
\bibitem[{Hedlund(1969)}]{hedlund1969endomorphisms}
\bibinfo{author}{Hedlund, G.A.}, \bibinfo{year}{1969}.
\newblock \bibinfo{title}{Endomorphisms and automorphisms of the shift
  dynamical system}.
\newblock \bibinfo{journal}{Mathematical systems theory} \bibinfo{volume}{3},
  \bibinfo{pages}{320--375}.
\bibitem[{Kari(2018)}]{Kari2018ReversibleCA}
\bibinfo{author}{Kari, J.}, \bibinfo{year}{2018}.
\newblock \bibinfo{title}{Reversible cellular automata: From fundamental
  classical results to recent developments}.
\newblock \bibinfo{journal}{New Generation Computing} \bibinfo{volume}{36},
  \bibinfo{pages}{145 -- 172}.
\newblock \URLprefix \url{https://api.semanticscholar.org/CorpusID:253842753}.
\bibitem[{Lafe(1997)}]{lafe1997data}
\bibinfo{author}{Lafe, O.}, \bibinfo{year}{1997}.
\newblock \bibinfo{title}{Data compression and encryption using cellular
  automata transforms}.
\newblock \bibinfo{journal}{Engineering Applications of Artificial
  Intelligence} \bibinfo{volume}{10}, \bibinfo{pages}{581--591}.
\bibitem[{Maruoka and Kimura(1976)}]{maruoka1976condition}
\bibinfo{author}{Maruoka, A.}, \bibinfo{author}{Kimura, M.},
  \bibinfo{year}{1976}.
\newblock \bibinfo{title}{Condition for injectivity of global maps for
  tessellation automata}.
\newblock \bibinfo{journal}{Information and Control} \bibinfo{volume}{32},
  \bibinfo{pages}{158--162}.
\bibitem[{Nandi et~al.(1994)Nandi, Kar and Chaudhuri}]{nandi1994theory}
\bibinfo{author}{Nandi, S.}, \bibinfo{author}{Kar, B.K.},
  \bibinfo{author}{Chaudhuri, P.P.}, \bibinfo{year}{1994}.
\newblock \bibinfo{title}{Theory and applications of cellular automata in
  cryptography}.
\newblock \bibinfo{journal}{IEEE Transactions on computers}
  \bibinfo{volume}{43}, \bibinfo{pages}{1346--1357}.
\bibitem[{Neumann(1966)}]{neumann1966theory}
\bibinfo{author}{Neumann, J.v.}, \bibinfo{year}{1966}.
\newblock \bibinfo{title}{Theory of self-reproducing automata}.
\newblock \bibinfo{journal}{Edited by Arthur W. Burks} .
\bibitem[{{Seck Tuoh Mora} et~al.(2005){Seck Tuoh Mora}, {Chapa Vergara},
  {Juárez Martínez} and McIntosh}]{SECKTUOHMORA2005134}
\bibinfo{author}{{Seck Tuoh Mora}, J.C.}, \bibinfo{author}{{Chapa Vergara},
  S.V.}, \bibinfo{author}{{Juárez Martínez}, G.}, \bibinfo{author}{McIntosh,
  H.V.}, \bibinfo{year}{2005}.
\newblock \bibinfo{title}{Procedures for calculating reversible one-dimensional
  cellular automata}.
\newblock \bibinfo{journal}{Physica D: Nonlinear Phenomena}
  \bibinfo{volume}{202}, \bibinfo{pages}{134--141}.
\newblock \URLprefix
  \url{https://www.sciencedirect.com/science/article/pii/S0167278905000540},
  \DOIprefix\doi{https://doi.org/10.1016/j.physd.2005.01.018}.
\bibitem[{Seck-Tuoh-Mora et~al.(2017)Seck-Tuoh-Mora, Medina-Marin,
  Hernandez-Romero, Martinez and Barragan-Vite}]{seck2017welch}
\bibinfo{author}{Seck-Tuoh-Mora, J.C.}, \bibinfo{author}{Medina-Marin, J.},
  \bibinfo{author}{Hernandez-Romero, N.}, \bibinfo{author}{Martinez, G.J.},
  \bibinfo{author}{Barragan-Vite, I.}, \bibinfo{year}{2017}.
\newblock \bibinfo{title}{Welch sets for random generation and representation
  of reversible one-dimensional cellular automata}.
\newblock \bibinfo{journal}{Information Sciences} \bibinfo{volume}{382},
  \bibinfo{pages}{81--95}.
\bibitem[{Sutner(1991)}]{sutner1991bruijn}
\bibinfo{author}{Sutner, K.}, \bibinfo{year}{1991}.
\newblock \bibinfo{title}{De bruijn graphs and linear cellular automata}.
\newblock \bibinfo{journal}{Complex Systems} \bibinfo{volume}{5},
  \bibinfo{pages}{19--30}.
\bibitem[{Tome(1994)}]{Tome1994}
\bibinfo{author}{Tome, J.}, \bibinfo{year}{1994}.
\newblock \bibinfo{title}{Necessary and sufficient conditions for reversibility
  in one dimensional cellular automata}, in: \bibinfo{booktitle}{Proceedings
  Workshop on Physics and Computation. PhysComp '94}, pp.
  \bibinfo{pages}{156--159}.
\newblock \DOIprefix\doi{10.1109/PHYCMP.1994.363686}.
\bibitem[{T{\'o}th and Lent(2001)}]{toth2001quantum}
\bibinfo{author}{T{\'o}th, G.}, \bibinfo{author}{Lent, C.S.},
  \bibinfo{year}{2001}.
\newblock \bibinfo{title}{Quantum computing with quantum-dot cellular
  automata}.
\newblock \bibinfo{journal}{Physical Review A} \bibinfo{volume}{63},
  \bibinfo{pages}{052315}.
\bibitem[{Wolfram(1983)}]{wolfram1983statistical}
\bibinfo{author}{Wolfram, S.}, \bibinfo{year}{1983}.
\newblock \bibinfo{title}{Statistical mechanics of cellular automata}.
\newblock \bibinfo{journal}{Reviews of modern physics} \bibinfo{volume}{55},
  \bibinfo{pages}{601}.
\bibitem[{Wolnik et~al.(2022)Wolnik, Dziemiańczuk, Dzedzej and {De
  Baets}}]{WOLNIK2022133075}
\bibinfo{author}{Wolnik, B.}, \bibinfo{author}{Dziemiańczuk, M.},
  \bibinfo{author}{Dzedzej, A.}, \bibinfo{author}{{De Baets}, B.},
  \bibinfo{year}{2022}.
\newblock \bibinfo{title}{Reversibility of number-conserving 1d cellular
  automata: Unlocking insights into the dynamics for larger state sets}.
\newblock \bibinfo{journal}{Physica D: Nonlinear Phenomena}
  \bibinfo{volume}{429}, \bibinfo{pages}{133075}.
\newblock \URLprefix
  \url{https://www.sciencedirect.com/science/article/pii/S0167278921002323},
  \DOIprefix\doi{https://doi.org/10.1016/j.physd.2021.133075}.
\bibitem[{Yang et~al.(2015)Yang, Wang and Xiang}]{yang2015reversibility}
\bibinfo{author}{Yang, B.}, \bibinfo{author}{Wang, C.}, \bibinfo{author}{Xiang,
  A.}, \bibinfo{year}{2015}.
\newblock \bibinfo{title}{Reversibility of general 1d linear cellular automata
  over the binary field z2 under null boundary conditions}.
\newblock \bibinfo{journal}{Information Sciences} \bibinfo{volume}{324},
  \bibinfo{pages}{23--31}.

\end{thebibliography}







\end{document}